\documentclass[envcountsame]{llncs}

\usepackage{tikz}
\usetikzlibrary{decorations.pathreplacing}
\usepackage{amsmath}
\usepackage{amssymb}

\newcommand{\fig}[3] 
{\begin{figure}[htbp]
 \begin{center}
 \end{center}
 \caption{#3}
 \label{#2}
 \end{figure}
}

\newcounter{pseudocode}

\title{\Large Improved Online Square-into-Square Packing \thanks{This work has been supported partially by NSF award CCF-1218620. I also want to thank my mentor Prof.\ Gandhi for all his
support.}
}

\author{Brian Brubach} 

\institute{
\email{brian.brubach@gmail.com}
}

\begin{document}

\maketitle

\begin{abstract} 
In this paper, we show an improved bound and new algorithm for the online square-into-square packing problem. This two-dimensional packing problem involves packing an online sequence of squares into a unit square container without any two squares overlapping. The goal is to find the largest area $\alpha$ such that any set of squares with total area $\alpha$ can be packed. We show an algorithm that can pack any set of squares with total area $\alpha \leq 3/8$ into a unit square in an online setting, improving the previous bound of $11/32$.
\end{abstract}

\section{Introduction}

In packing problems, we wish to place a set of objects into a container such that no two objects overlap. These problems have been studied extensively and have numerous applications. However, even common one-dimensional versions of packing problems, such as the Knapsack problem, are NP-hard. In these difficult problems, it is often important to know whether it is even feasible to pack a given set into a particular container. In fact, for the two-dimensional case, it is worth noting that merely checking whether a given set of squares can be packed into a unit square was shown to be NP-hard by Leung, et al.\ \cite{L90}. 

The square-into-square packing problem asks, ``What is the largest area $\alpha$ such that any set of squares with total area $\alpha$ can be packed into a unit square without overlapping?'' It is trivial to show the upper bound that $\alpha \leq 1/2$. Two squares of height $1/2 + \epsilon$ cannot be packed into a unit square container. In addition, Moon and Moser \cite{MM67} showed in 1967 that the bound of $1/2$ is tight in the offline case. Squares can be sorted in decreasing order and packed from left-to-right into horizontal ``shelves'' starting along the bottom of the container. The height of each shelf is determined by the largest object in the shelf and when one shelf fills, a new shelf is opened directly above it.

In the online version of the problem, we have no knowledge of the full set of squares to be packed. Squares are received one at a time and each must be packed before seeing the next square. Once a square is packed, it cannot be moved. In this case, the successful offline approach cannot be used as it requires sorting the set. The current best lower bound is $\alpha \geq 11/32$ by Fekete and Hoffmann \cite{FH13} in 2013. They take a dynamic, multi-directional shelf-packing approach that combines horizontal shelves with dynamically allocated vertical shelves.

\subsection{Related work}

\noindent
\textbf{Offline Square Packing. }
Early related work involved packing a set of objects into the smallest possible rectangle container. Moser \cite{M66} posed the following question in 1966: ``What is the smallest number $A$ such that any family of objects with total area at most $1$ can be packed into a rectangle of area $A$?'' Since then, there have been many results for the offline packing of squares into rectangle containers.

In 1967, Moon and Moser \cite{MM67} showed that any set of squares with total area $1$ can be packed into a square of height $\sqrt{2}$. This established $A \leq 2$ or in the terms of our problem, $\alpha \geq 1/2$. This result was followed by several improvements on the value of $A$ using rectangular containers. The current best upper bound is $1.3999$ by Hougardy \cite{H11} in 2011. \\

\noindent
\textbf{Online Square Packing. }
In 1997, Januszewski and Lassak \cite{JL97} considered the online variant in many dimensions. For two-dimensional square-into-square packing, their work showed a bound of $\alpha \geq 5/16$ by recursively dividing a unit square container into rectangles of aspect ratio $\sqrt{2}$. In 2008, Han et al. \cite{H08} used a similar approach to improve the lower bound to $1/3$.

Fekete and Hoffmann \cite{FH13} provided a new approach in 2013 which uses multi-directional shelves (horizontal and vertical) that are allocated dynamically. Using this technique, they were able to improve the lower bound further to $11/32$, the current best.

\subsection{Our Contributions}

We show a new lower bound of $\alpha \geq 3/8$ for online square-into-square packing, improving upon the previous result of $11/32$. The algorithm we designed uses dynamically allocated, multi-directional shelves like Fekete and Hoffmann \cite{FH13}. However, our approach differs in several important ways besides achieving a stronger lower bound.

We propose new criteria for packing shelves that allows us to make stronger claims about the density of individual shelves and simplifies the analysis. We also dramatically simplify the use of \textit{buffer regions}\footnote{Buffer regions are sections of the container which are used in the analysis to support claims made about other sections. For instance, let region $A$ and buffer region $B$ be two disjoint sections within the container. We may wish to guarantee that the total area of squares packed into region $A$ is $1/2$ the area of $A$. To do so, we may count both squares in $A$ as well as squares in $B$ which we have \textit{assigned} to $A$.} and describe a new way to partition the container into four primary shelves.  All of these improvements may be of independent interest for other 2-dimensional packing problems.

\subsection{Outline}
In section $2$, we discuss preliminaries including terminology and notation. In section $3$, we present our simple algorithm for packing squares into a unit square container. In section $4$, we analyze our algorithm and show that it successfully packs any set of squares with total area at most $3/8$.

\section{Preliminaries}

\subsection{Terminology}
In this algorithm, we define a \textit{shelf} $S$ as a subrectangle in the container with height $h$, length $\ell$, and packing ratio $r$, $0 < r < 1$. The \textit{packing ratio} is the ratio of the smallest possible height to the largest possible height of squares that can be packed into $S$. Any square packed into $S$ must have height $k$, $h \geq k > hr$.

When \textit{packing} a shelf $S$, squares are added side-by-side to $S$. In our algorithm, we also pack small vertical shelves into larger horizontal shelves. These vertical shelves are also packed side-by-side with squares and other vertical shelves.

\begin{figure*}
\begin{tikzpicture}[x=10cm,y=10cm]

\draw (0.5, 0.3) node {\large{Shelf Packing}};

\filldraw[fill=black!15!white]
(13/80, 0) rectangle (23/80,1/4)
(34/80, 0) rectangle (34/80 + 0.0577,  1/4)
(34/80 + 0.0577,  0) rectangle (34/80 + 0.0577 + 0.0577*0.58,  1/4);

\filldraw[fill=black!35!white] 

(0, 0) rectangle(13/80, 13/80) 

(15/80, 0) rectangle (23/80, 8/80) (17/80, 8/80) rectangle(23/80, 14/80) 

(23/80, 0) rectangle (34/80,11/80) 

(34/80 + 0.0577 - 0.035, 0) rectangle (34/80 + 0.0577,  0.035) 
(34/80 + 0.0577 - 0.045, 0.035) rectangle (34/80 + 0.0577,  0.08) 
(34/80 + 0.0577 - 0.04, 0.08) rectangle (34/80 + 0.0577,  0.12) 

(34/80 + 0.0577 + 0.0577*0.58 - 0.04*0.58, 0) rectangle (34/80 + 0.0577 + 0.0577*0.58,  0.04*0.58) 

(34/80 + 0.0577 + 0.0577*0.58, 0) rectangle (52/80 + 0.0577 + 0.0577*0.58,18/80) 
(52/80 + 0.0577 + 0.0577*0.58, 0) rectangle (68/80 + 0.0577 + 0.0577*0.58,14/80);







%
%


\draw (-0.025, 0) -- (-0.075, 0) (-0.025, 1/8) -- (-0.075, 1/8);
\draw[<->] (-0.0375, 0) -- (-0.0375, 1/8);
\draw (-0.07, 1/16) node {$hr$};

\draw (-0.1, 0) -- (-0.15, 0) (-0.1, 1/4) -- (-0.15, 1/4);
\draw[<->] (-0.1125, 0) -- (-0.1125, 1/4);
\draw (-0.1375, 1/8) node {$h$};


\draw (0, -0.025) -- (0, -0.075) (1, -0.025) -- (1, -0.075);
\draw[<->] (0, -0.0375) -- (1, -0.0375);
\draw (1/2, -0.07) node {$\ell$};

\draw[very thick] (0,0) rectangle (1,1/4);


%


\end{tikzpicture}
\begin{center}
\caption{Illustration of shelf packing with $r = 1/2$ for the horizontal shelf. Vertical shelves are designated with a light gray background and have different packing ratios.}
\label{fig:shelf}
\end{center}
\end{figure*}
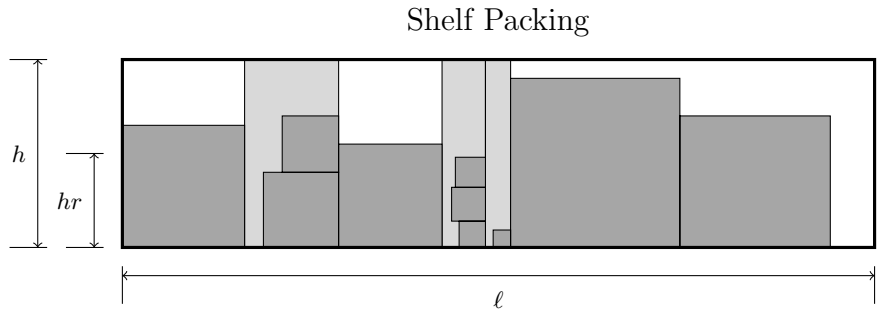

In addition, shelves may be considered \textit{open} or \textit{closed}. A shelf $S$ is initially considered open. As squares are added to $S$, we may receive a square $Q$ with height $h_Q$, $h \geq h_Q > hr$, such that packing $Q$ into $S$ would exceed the length of $S$. At this point, we say that $S$ is closed and never attempt to pack any future squares into $S$. The new square $Q$ is then packed into some other shelf.

In the analysis, we refer to the total area of all squares packed into a shelf or other section of the container as the \textit{covered area} of that section. We also refer to \textit{assigned covered area} or covered area \textit{assigned} to a section. This assigned covered area may include squares from elsewhere in the contained which have had their areas assigned to this section for the purpose of clearer analysis. A single square may have parts of its area assigned to different sections as long as the sum of those parts is less than or equal to the area of the square itself. We use the term \textit{density} to describe the ratio of the assigned covered area of section to the total area of that section.

Specifically when analyzing shelves, we refer to the \textit{used length} of a shelf to describe the length of that shelf which is occupied by squares or vertical shelves. In other words, this represents the length of the shelf which may not be overlapped by any future squares.

\subsection{Size Classes of Squares} \label{sec:class}
We divide possible input squares into four classes based on height: large, medium, small, and very small.
\begin {itemize}
\item {\textbf {Large: }} height $> \frac{1}{2}$
\item {\textbf {Medium: }} height $\leq \frac{1}{2}$ and $> \frac{1}{4}$
\item {\textbf {Small: }} height $\leq \frac{1}{4}$ and $> \frac{1}{8}$
\item {\textbf {Very Small: }} height $\leq \frac{1}{8}$
\end{itemize}

We also refer to small squares as class $c_0$ and subdivide the very small squares into subclasses $c_i$, $i \geq 1$. Squares in $c_i$ are packed into shelves with max height $h_i$ and packing ratio $r_i$. Naturally, squares in $c_i$ will have height $k_i$, $h_i \geq k_i > h_i r_i$. We use the notation $c_{j+}$ to refer to all $c_i$, $i \geq j$.

In our algorithm, we assign ratios as follows: $r_0 = 0.5$, $r_1 = 0.71$, $r_2 = 0.65$, and $r_{3+} = 0.58$. To account for all small and very small squares with height $\leq \frac{1}{4}$, we note that $h_0 = \frac{1}{4}$ and for all $i \geq 1$, we set $h_i = h_{i-1}r_{i-1}$. In section~\ref{sec:smallveryshelves}, we will show that our choices for $r_{1+}$ ensure certain desirable properties when we pack vertical shelves for very small squares into horizontal shelves for small squares. Most notably, closed vertical shelves for very small squares will have a density greater than $0.5$, which is the packing ratio for small squares.

\section{Algorithm}

For each square, we pack it according to a subroutine based on its size:
\begin {itemize}
\item {\textbf {Large: }} A large square is packed into the upper right corner of the container.
\item {\textbf {Medium: }} Medium squares are first packed from right-to-left along the bottom of the container until we receive a medium square that does not fit into this space. Then, they are packed from right-to-left along the top of the container.
\item {\textbf {Small: }} Small squares are packed in three parts. Initially, we pack the small buffer $b_0$ until it is closed. Then, we alternate packing the primary shelves $p_1$ and $p_2$, choosing the shelf with the shortest used length until both are closed. Finally, we alternate packing the primary shelves $p_3$ and $p_4$ in the same way.
\item {\textbf {Very Small: }} Very small squares are packed based on their subclass $c_i$, $i \geq 1$. For each subclass, we maintain exactly one open vertical shelf at any given time. Initially, all of these open vertical shelves are in the buffer regions, labeled $b_1$, $b_2$, $b_3$...  in Fig.~\ref{fig:Over}. Whenever a vertical shelf is closed, we open a new vertical shelf and the vertical shelf itself is ``packed'' into the container as if it is a small square. 
\end{itemize}

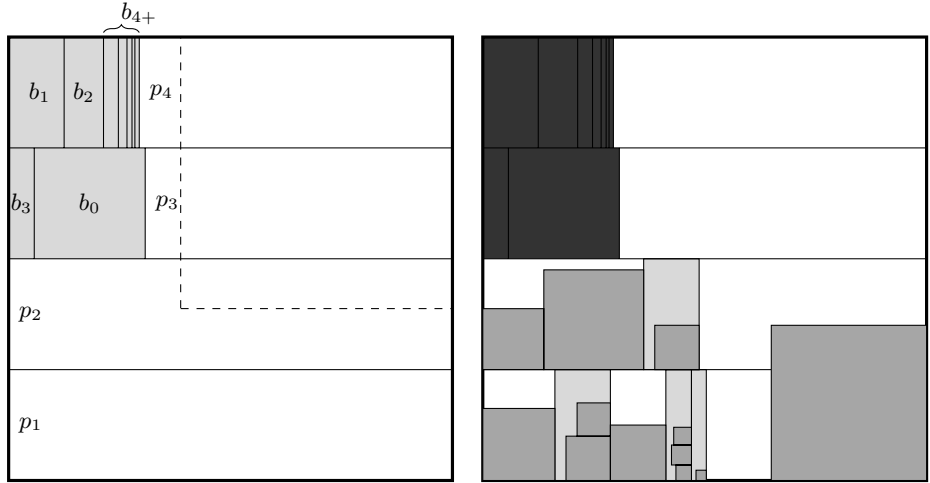
\begin{figure*}
\begin{tikzpicture}[x=5.9cm,y=5.9cm]

\fill[fill=black!15!white] 
(0, 3/4) rectangle(0.2943, 1) 
(0, 1/2) rectangle(0.0577 + 1/4, 3/4);

\draw
(0,1/4) -- (1,1/4) (0,2/4) -- (1,2/4) (0,3/4) -- (1,3/4)

(0.0577,1/2) -- (0.0577,3/4)
(0.0577 + 1/4, 1/2) -- (0.0577 + 1/4, 3/4)
(1/8,3/4) -- (1/8,1) 
(1/8 + 0.08875, 3/4) -- (1/8 + 0.08875, 1) 
(1/8 + 0.08875 + 0.0335, 3/4) -- (1/8 + 0.08875 + 0.0335, 1) 
(1/8 + 0.08875 + 0.0335 + 0.0194, 3/4) -- (1/8 + 0.08875 + 0.0335 + 0.0194, 1)
(1/8 + 0.08875 + 0.0335 + 0.0194 + 0.0113, 3/4) -- (1/8 + 0.08875 + 0.0335 + 0.0194 + 0.0113, 1) 
(1/8 + 0.08875 + 0.0335 + 0.0194 + 0.0113 + 0.0065, 3/4) -- (1/8 + 0.08875 + 0.0335 + 0.0194 + 0.0113 + 0.0065, 1) 
(0.2943, 3/4) -- (0.2943, 1);

\draw 
(0.05, 1/8) node {$p_1$} 
(0.05, 3/8) node {$p_2$}
(0.05 + 0.0577 + 1/4, 5/8) node {$p_3$} 
(0.05 + 0.2943, 7/8) node {$p_4$};

\draw 
(0.0577 + 1/8, 5/8) node {$b_0$} 
(0.0577/2, 5/8) node {$b_3$}
(0.07, 7/8) node {$b_1$}
(1/16 + 0.214/2, 7/8) node {$b_2$};
\draw [decorate,decoration={brace, amplitude=3pt}]
(1/8 + 0.08875, 1.01) -- (0.2943, 1.01) node [above] {$b_{4+}$};



\draw[dashed] (1-0.61237, 1-0.61237) rectangle(1,1);

\draw[very thick] (0,0) rectangle (1,1);

\end{tikzpicture} ~~~
\begin{tikzpicture}[x=5.9cm,y=5.9cm]

\fill[fill=black!80!white]
(0, 3/4) rectangle(0.2943, 1) 
(0, 1/2) rectangle(0.0577 + 1/4, 3/4);




\draw[very thick] (0,0) rectangle (1,1);


\filldraw[fill=black!15!white]
(13/80, 0) rectangle (23/80,1/4) 
(33/80, 0) rectangle (33/80 + 0.0577,  1/4) 
(33/80 + 0.0577,  0) rectangle (33/80 + 0.0577 + 0.0577*0.58,  1/4) 
(29/80, 1/4) rectangle (39/80,1/2); 

\filldraw[fill=black!35!white] (0, 0) rectangle(13/80, 13/80) 
(15/80, 0) rectangle (23/80, 8/80) (17/80, 8/80) rectangle(23/80, 14/80) 
(23/80, 0) rectangle (33/80,10/80) 
(33/80 + 0.0577 - 0.035, 0) rectangle (33/80 + 0.0577,  0.035) 
(33/80 + 0.0577 - 0.045, 0.035) rectangle (33/80 + 0.0577,  0.08) 
(33/80 + 0.0577 - 0.04, 0.08) rectangle (33/80 + 0.0577,  0.12) 
(33/80 + 0.0577 + 0.0577*0.58 - 0.04*0.58, 0) rectangle (33/80 + 0.0577 + 0.0577*0.58,  0.04*0.58); 

\filldraw[fill=black!35!white] (0, 1/4) rectangle(11/80, 31/80) 
(11/80, 1/4) rectangle (29/80, 38/80) 
(31/80, 1/4) rectangle (39/80, 28/80); 

\draw
(0,1/4) -- (1,1/4) (0,2/4) -- (1,2/4) (0,3/4) -- (1,3/4)

(0.0577,1/2) -- (0.0577,3/4)
(0.0577 + 1/4, 1/2) -- (0.0577 + 1/4, 3/4)
(1/8,3/4) -- (1/8,1) 
(1/8 + 0.08875, 3/4) -- (1/8 + 0.08875, 1) 
(1/8 + 0.08875 + 0.0335, 3/4) -- (1/8 + 0.08875 + 0.0335, 1) 
(1/8 + 0.08875 + 0.0335 + 0.0194, 3/4) -- (1/8 + 0.08875 + 0.0335 + 0.0194, 1)
(1/8 + 0.08875 + 0.0335 + 0.0194 + 0.0113, 3/4) -- (1/8 + 0.08875 + 0.0335 + 0.0194 + 0.0113, 1) 
(1/8 + 0.08875 + 0.0335 + 0.0194 + 0.0113 + 0.0065, 3/4) -- (1/8 + 0.08875 + 0.0335 + 0.0194 + 0.0113 + 0.0065, 1) 
(0.2943, 3/4) -- (0.2943, 1);


\filldraw[fill=black!35!white] (1 - 28/80, 0) rectangle (1, 28/80); 


\end{tikzpicture}
\begin{center}
\caption{Left: Overview of shelf structure. The four primary shelves are labeled $p_1$ through $p_4$.  The buffer region containing $b_0$, $b_1$, $b_2$, etc., is designated by the shaded area in the upper left corner. The region inside the dashed line in the upper left corner illustrates the size of the largest possible square we could receive. 
Right: Example of packing medium, small, and very small squares into the bottom half of the container}
\label{fig:Over}
\end{center}
\end{figure*}

\section{Analysis}

The analysis is broken down into several parts. Section~\ref{sec:shelf} covers a basic lemma and corollary about 2-dimensional shelf packing. Section~\ref{sec:largemed} shows that large and medium squares with total area at most $\leq 3/8$ can be packed by our algorithm. Sections~\ref{sec:smallveryshelves},~\ref{sec:buffer}, and~\ref{sec:smallverypack} are devoted to small and very small squares. The first two cover vertical shelves and buffer regions, respectively. Section~\ref{sec:smallverypack} shows that small and very small squares with total area at most $\leq 3/8$ can be packed by our algorithm. Section~\ref{sec:combine} combines the results of the previous sections to show our main result.
In this final section, we show that any set of squares with total area at most $\leq 3/8$, which is received in an online fashion, can be packed into a unit square container.

\subsection{Shelf Packing}\label{sec:shelf}

Recall that we define a shelf $S$ as a subrectangle in the container with height $h$, length $\ell$, and packing ratio $r$, $0 < r < 1$. Squares are added side-by-side to $S$ and any square with height $k$, $h \geq k > hr$, can be added to $S$ provided there is room. \\

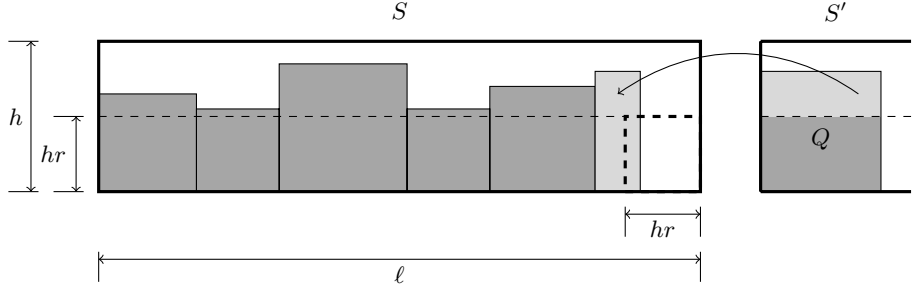
\begin{figure*}

\begin{tikzpicture}[x=8cm,y=8cm]

\filldraw[fill=black!35!white] 
(0, 0) rectangle(13/80, 13/80) 
(13/80, 0) rectangle(24/80,11/80)
(24/80, 0) rectangle(41/80,17/80) 
(41/80, 0) rectangle(52/80,11/80) 
(52/80, 0) rectangle(66/80,14/80);

\filldraw[fill=black!15!white] (66/80, 0) rectangle(72/80,16/80);

\draw[dashed] (0,1/8) -- (7/8,1/8);

\draw[very thick, dashed] (7/8,0) rectangle (1,1/8);

\draw (0.5, 0.3) node {$S$};


\draw[very thick] (0,0) rectangle (1,1/4);


\fill[fill=black!35!white, shift={(1.1,0)}] (0, 0) rectangle(16/80, 10/80);
\fill[fill=black!15!white, shift={(1.1,0)}] (0, 10/80) rectangle(16/80, 16/80);
\draw[shift={(1.1,0)}] (0, 0) rectangle(16/80, 16/80);

\draw[dashed, shift={(1.1,0)}] (0,1/8) -- (1/4,1/8);

\draw[very thick, shift={(1.1,0)}] (0,0) -- (1/4,0) (0,1/4) -- (1/4,1/4) (0,0) -- (0,1/4);

\draw (1.1 + 1/8, 0.3) node {$S'$};
\draw (1.1 + 8/80, 7/80) node {$Q$};

\draw[->] (1.1 + 13/80,13/80) to [out=145,in=35] (69/80, 13/80);


\draw (-0.025, 0) -- (-0.075, 0) (-0.025, 1/8) -- (-0.075, 1/8);
\draw[<->] (-0.0375, 0) -- (-0.0375, 1/8);
\draw (-0.075, 1/16) node {$hr$};

\draw (-0.1, 0) -- (-0.15, 0) (-0.1, 1/4) -- (-0.15, 1/4);
\draw[<->] (-0.1125, 0) -- (-0.1125, 1/4);
\draw (-0.1375, 1/8) node {$h$};

\draw (7/8, -0.025) -- (7/8, -0.075) (1, -0.025) -- (1, -0.075);
\draw[<->] (7/8, -0.0375) -- (1, -0.0375);
\draw (15/16, -0.0625) node {$hr$};

\draw (0, -0.1) -- (0, -0.15) (1, -0.1) -- (1, -0.15);
\draw[<->] (0, -0.1125) -- (1, -0.1125);
\draw (1/2, -0.1375) node {$\ell$};

\end{tikzpicture}
\begin{center}
\caption{Illustration of Lemma 1 and Corollary 1 for the case when $r = 1/2$. The upper portion of the square $Q$ is assigned to $S$, while the lower portion is assigned to $S'$.}
\label{fig:Lem1}
\end{center}
\end{figure*}

The following lemma is a generalization of a lemma due to Moon and Moser \cite{MM67}. Their lemma applies when $r = h/2$. Here, we find it useful to consider any $r \in (0, 1)$.

\begin{lemma}
\label{lemma:MM}
Let $S$ be a shelf with height $h$, length $\ell$ and packing ratio $r$, $0 < r < 1$, that is packed with a set $P$ of squares with height $\leq h$ and $> hr$. Let $Q$ be an additional square with height $h_Q$, $h \geq h_Q > hr$, that does not fit into $S$. Then the total area of all the squares packed into $S$ plus the area of $Q$ is greater than $\ell hr - (hr)^2 + h_Qhr$.
\end{lemma}
\begin{proof}
Given that $Q$ does not fit into $S$, it must be the case that the set $P$ covers an area of at least $hr(\ell -h_Q) = \ell hr-h_Qhr$. The area of $Q$ is clearly $h_Q^2$. Combining the two, we get a covered area of at least,
\begin{align*}
\ell hr - h_Qhr+ h_Q^2 &= \ell hr - h_Qhr + h_Q^2 + (hr)^2 - (hr)^2 + h_Qhr - h_Qhr \\
&= \ell hr - (hr)^2 + h_Qhr + h_Q^2 - 2h_Qhr + (hr)^2 \\
&= \ell hr - (hr)^2 + h_Qhr + (h_Q - hr)^2 \\
&> \ell hr - (hr)^2 + h_Qhr &\because h_Q > hr
\end{align*}
\end{proof}

This leads to the following corollary which we will use to bound the covered area assigned to closed shelves.

\begin{corollary}
\label{cor:MM}
We can assign an area of $\ell hr - (hr)^2$ to a closed shelf $S$.
\end{corollary}
\begin{proof}
By Lemma~\ref{lemma:MM} the covered area of $S$ plus the area of $Q$ is greater than $\ell hr - (hr)^2 +  h_Qhr$. Let $S'$ be the open shelf holding $Q$. We can assign the area $\ell hr - (hr)^2$ to $S$ and let $h_Qhr$ be assigned to $S'$.
\end{proof}

\subsection{Packing Large and Medium Squares}\label{sec:largemed}

We will first show that any input containing only large or only medium squares can be packed. Then we will address inputs containing both size classes. \\

\noindent
\textbf{Large Squares:  }
Large squares have height $> 1/2$ and so cover an area $> 1/4$. Because the total area to be covered is at most $3/8$, there can be at most one large square in the input. Additionally, notice that when the covered area of the container exceeds $1/8$, large squares are no longer possible. We will ensure there is room in the upper right corner for the largest possible square in the remaining input until the covered area exceeds $1/8$. The dashed line in Fig.~\ref{fig:Over} shows the initial space for a large square with an area of $3/8$. \\

\noindent
\textbf{Medium Squares:  }
Medium squares have height $\leq 1/2$ and $> 1/4$. As such, each must cover an area $> 1/16$ and there can be at most $5$ medium squares in the input. To analyze the packing of these squares, we may think of the combined space of $p_1$ and $p_2$ in Fig.~\ref{fig:Over} as the \textit{bottom shelf} for medium squares. This bottom shelf has length $1$, height $1/2$, and packing ratio $1/2$. Likewise, we may think of the combined space of $p_3$ and $p_4$ as the \textit{top shelf} with length $\geq 0.692$ (the length of $p_3$ as defined in section~\ref{sec:buffer} below), height $1/2$, and packing ratio $1/2$. Using this terminology, our algorithm first packs medium squares from right-to-left into the bottom shelf. Then, they are packed from right-to-left into the top shelf. 

The following lemma will show that this accommodates any input containing only medium squares without violating the buffer regions in the upper left.

\begin{lemma}
\label{lemma:med}
Any set of medium squares with total area $\leq 3/8$ can be packed into the bottom and top shelves.
\end{lemma}
\begin{proof} 
We prove by contradiction that up to five medium squares can be packed into these shelves without overlapping. In particular, we consider two cases: $3$ medium squares in the top shelf and $2$ medium squares in the top shelf. \\

\noindent
\textbf{Case 1:}
There are $3$ medium squares in the top shelf. In this case, there must be some square $Q$ with height $h_Q$ which closes the bottom shelf and is packed into the top shelf. Following $Q$, two more medium squares must be placed on the top shelf and all three of these squares must have height $> 1/4$. Then, the total area of the bottom shelf squares plus the three squares on the top shelf is greater than 
\begin{align*}
1/4(1 - h_Q) + h_Q^2  + (1/4)^2 + (1/4)^2 = 3/8 + h_Q^2 - (1/4 \cdot h_Q)
\end{align*}
which is $> 3/8$ for $h_Q > 1/4$, a contradiction. \\

\noindent
\textbf{Case 2:}
There are $2$ medium squares in the top shelf. In this case, there must be some square $Q$ with height $h_Q$ which closes the bottom shelf and is packed into the top shelf. Following $Q$, another medium square, $Q_2$ with height $h_{Q_2}$, must also be packed into the top shelf. In order to extend beyond the length of the top shelf, the combined height of these two squares must be $> 0.692 > 0.69$. So we can say that $h_{Q_2} \geq 0.69 - h_Q$. Then, the total area of the bottom shelf squares plus the two squares on the top shelf is greater than 

\begin{align*}
1/4(1 - h_Q) + h_Q^2  + (0.69 - h_Q)^2 	&= 0.25 - 0.25h_Q + h_Q^2 \\
								&~~~+ 0.4761 - 1.38h_Q + h_Q^2 \\
&= 0.7261 - 1.63h_Q + 2h_Q^2
\end{align*}
which has a minimum value $> 3/8$, a contradiction. \\
\end{proof}

\noindent
\textbf{Medium and Large Squares:  }
We use the following lemma to show that an input containing both medium and large squares will not be a problem.

\begin{lemma}
\label{lemma:medlarge}
For inputs with total area $\leq 3/8$, large and medium squares will never overlap.
\end{lemma}

\begin{proof} 
We first note that any input containing only these two size classes must contain exactly $1$ large square packed into the upper right corner and $1$ medium square packed into the lower left corner. These two squares cannot overlap because the combined height of any two squares with total area $\frac{3}{8}$ cannot exceed $2 \sqrt{3/8 \cdot 1/2} \approx 0.87$ which is less than the height of the container.
\end{proof}

\subsection{Small and Very Small Shelves}\label{sec:smallveryshelves}

Refer back to section~\ref{sec:class} for an overview of how we subdivide the small and very small classes. Recall that each class $c_i$, $i \geq 0$, represents squares that are packed into shelves with max height $h_i$ and packing ratio $r_i$. It follows that squares in $c_i$ must have minimum height $> h_i r_i$. To account for all squares with height $\leq \frac{1}{4}$, we assign $h_0 = 1/4$ and for all $i \geq 1$, $h_i = h_{i-1}r_{i-1}$.

According to our algorithm, $c_0$ squares are packed from left to right into the four horizontal primary shelves ($p_1$ to $p_4$) with length $1$, height $1/4$, and packing ratio $1/2$. For all $k \geq 1$, $c_{k}$ squares are packed into vertical shelves of length $\frac{1}{4}$, height $h_k$, and packing ratio $r_k$. These vertical shelves are also added to the primary shelves from left to right. We further enforce a rule for vertical shelves that each class $c_k$ may have at most one open shelf at any given time.

A new vertical shelf is opened only when another must be closed. The buffer regions $b_0$, $b_1$, $b_2$, ..., in the upper left corner of the container are packed first and allow us to assign extra covered area when we need it.

As stated in our algorithm, the packing of small squares happens in three parts. First, we fill the buffer regions. Second, we alternate filling $p_1$ and $p_2$. Third, we alternate filling $p_3$ and $p_4$. Before showing that we can pack small and very small squares, we will examine vertical and primary shelves in greater detail. \\

\noindent
\textbf{Bounding Density of Vertical Shelves:  }
First, we show how heights and ratios are chosen for small and very small classes. For primary shelves with height $h_0$, one goal is to ensure that any used length $\ell_u$ can be assigned a covered area of at least $\frac{h_0\ell_u}{2}$. In other words, the used length of a primary shelf will have a density of $1/2$.

If we choose $r_0$ to be $\frac{1}{2}$, any used section of a small shelf containing only $c_0$ squares will clearly have half of its area covered. Naturally, it would be useful if closed vertical shelves supplied us with the same guarantee. To accomplish this, we carefully choose values for $r_1$, $r_2$, and $r_{3+}$. \\

We first consider $c_1$. Vertical shelves of this class have height $h_1$ and length $\ell = h_0 = 2h_1$. In this special case, we show that any packing ratio $r_1 \geq \sqrt{1/2}$ suffices.

\begin{lemma}
\label{lemma:c1}
Let $S$ be a shelf with height $h$, length $\ell$ and packing ratio $r$, $0 < r < 1$. If $\ell = 2h$, we can choose a packing ratio $r \geq \sqrt{1/2}$, such that the covered area assigned to $S$ is at least $\frac{\ell h}{2}$ when $S$ is closed.
\end{lemma}
\begin{proof}
Notice that if $\ell = 2h$ and $r > 2/3$, any closed shelf $S$ must contain exactly two squares. Then if $r \geq \sqrt{1/2} > 2/3$, the minimum packing density of a closed shelf $S$ is
\begin{align*}
\frac{2(hr)^2}{\ell h} \geq \frac{2h^2(\sqrt{1/2})^2}{2h^2} = \frac{1}{2}
\end{align*}
Then, for any $r \geq \sqrt{1/2}$, the total covered area is at least $\frac{\ell h}{2}$. \\
\end{proof}

For $c_{2+}$, we can rely on a weaker, more general claim because shelves have length $\ell \geq (2 + \epsilon)h$ for some $\epsilon > 0$.

\begin{lemma}
\label{lemma:c2}
Let $S$ be a shelf with height $h$, length $\ell \geq (2 + \epsilon)h$ and packing ratio $r$, $0 < r < 1$. We can choose a packing ratio $r$ satisfying $r - \frac{hr^2}{\ell} \geq \frac{1}{2}$, such that the covered area assigned to $S$ is at least $\frac{\ell h}{2}$ when $S$ is closed.
\end{lemma}
\begin{proof}
If $\ell \geq (2 + \epsilon)h$, then $r - \frac{hr^2}{\ell} \geq r - \frac{r^2}{2 + \epsilon}$ and we can choose an $r < 1$ such that $r - \frac{hr^2}{\ell} \geq \frac{1}{2}$. Multiplying through by $\ell h$ gives us $\ell hr - (hr)^2 \geq \frac{\ell h}{2}$. Using Corollary~\ref{cor:MM}, we can assign the area $\ell hr - (hr)^2$ to $S$ when it closes. So $S$ has been assigned a covered area of at least $\frac{\ell h}{2}$ when it is closed. \\
\end{proof}

\noindent
\textbf{Summary of Heights, Ratios, and Packing Densities:  } \\

\begin{tabular}{|c|c|c|}
  \hline
  Height & Ratio & Packing Density \\ \hline
  $h_0 = 0.25$ & $r_0 = 0.5$ & $> 0.5$ \\ \hline
  $h_1 = 0.125$ & $r_1 = 0.71$ & $0.71^2 > 0.5$ \\ \hline
  $h_2 = 0.08875$ & $r_2 = 0.65$ & $r - \frac{hr^2}{l} = 0.65 - \frac{0.08875*0.65^2}{0.25}  > 0.5$ \\ \hline
  $h_3 \approx 0.0577$ & $r_{3+} = 0.58$ & $r - \frac{hr^2}{l} > 0.58 - \frac{0.0577*0.58^2}{0.25}  > 0.5$ \\
  \hline
\end{tabular} \\ \\

Notice that $0.58$ is acceptable for all $r_{3+}$ since a decrease in height can only increase the density shown by Lemma~\ref{lemma:c2}.

\begin{lemma}
\label{lemma:vert}
Corollary~\ref{cor:MM} can be extended to primary horizontal shelves if all vertical shelves packed into them are closed.
\end{lemma}
\begin{proof}
Using Lemmas~\ref{lemma:c1} and~\ref{lemma:c2}, we can now say our packing ratios ensure that any closed vertical shelf packed into a primary horizontal shelf has been assigned a covered area greater than half of the area it uses. 

As in Lemma~\ref{lemma:MM}, let $S$ be a primary shelf packed with a set $P$ of some combination of small squares and vertical shelves. Let $Q$ be an additional square or vertical shelf that does not fit into $S$. Notice that $P$ must still cover an area of $hr(\ell -h_Q) = \ell hr-h_Qhr$

Clearly, if $Q$ is a small square, the proof in Lemma~\ref{lemma:MM} applies. Otherwise, if $Q$ is a vertical shelf, consider the fact that the largest vertical shelf has the height $h_1 = 1/8$. If such a shelf causes $S$ to close, then $\ell - h_Q > 7/8$ and $P$ must cover an area of at least $hr(7/8) = \ell hr - (hr)^2$
\end{proof}

\subsection{Buffer Regions}\label{sec:buffer}

\noindent
\textbf{Buffer Region for Very Small Squares ($b_{1+}$):  }
This region includes a vertical shelf for each subclass $c_k$, $k \geq 1$. Squares of these subclasses must fill their $b_k$ shelf before being packed elsewhere. However, we never assign covered area to a buffer region shelf. Instead, for every closed vertical shelf $S_k$ of $c_k$ squares in $b_k$, we assign its covered area to the open vertical shelf $S_k'$ elsewhere in the container. This allows us to treat all vertical shelves outside of the buffer region as closed shelves when calculating the assigned covered area of the primary shelves they are contained in. \\

\noindent
\textbf{Buffer Region for Small Squares ($b_0$):  }
This is the initial packing region for $c_0$ squares and $c_{1+}$ squares with closed shelves in $b_{1+}$. Covered area in $b_0$ will be assigned elsewhere in the same way as the $b_{1+}$ buffers. The length of $b_0$ is $1/4 = h_0$.

By Lemma~\ref{lemma:vert}, if $b_0$ receives only $c_0$ squares, we can guarantee an assigned covered area of at least $\ell hr - (hr)^2 = h_0^2r_0 - (h_0r_0)^2 =1/64 = (h_0r_0)^2$.

The consequence is that we have $(h_0r_0)^2$ of extra covered area to assign when packing $p_1$ and $p_2$. Additionally, the packing of small and very small squares into $p_3$ will begin in the unused portion of $b_0$. This ensures that $b_0$ is fully utilized and gives another $1/64 = (h_0r_0)^2$ to assign when packing such squares into $p_3$ and $p_4$. \\

\noindent
\textbf{Lengths of Buffer Regions and Primary Shelves 3 and 4:  }
The length of $b_0$ and $b_3$ is $< 0.308$. The length of $b_1$, $b_2$, and $b_{4+}$ is $< 0.294$ and can be computed as follows:
\begin{align*}
\sum_{i=1}^\infty h_i - h_3 &= h_1 + h_2 + \sum_{i=0}^\infty h_40.58^i &\because r_{3+} = 0.58 \\
&= h_0(r_0 + r_0r_1 + \sum_{i=0}^\infty r_0r_1r_2r_30.58^i) \\
&= h_0(r_0(1 + r_1(1 + r_2r_3 \sum_{i=0}^\infty 0.58^i))) \\
&= 0.25(0.5(1 + 0.71(1 + 0.65 \cdot 0.58 \sum_{i=0}^\infty 0.58^i))) \\
&= 0.25(0.5(1 + 0.71(1 + 0.65 \cdot 0.58 \cdot \frac{1}{1 - 0.58}))) \\
&< 0.294 \\
\end{align*}

 Therefore, the primary shelves $p_3$ and $p_4$ have lengths $\geq 0.692$ and $\geq 0.706$, respectively. \\

\subsection{Packing Small and Very Small Squares}\label{sec:smallverypack}
After initially filling $b_0$ and potentially some or all of $b_{1+}$, we continue packing in primary shelves $p_1$ and $p_2$ in alternating fashion. Each time we receive a $c_0$ square or open a new vertical shelf, we do so in whichever of these two primary shelves has the shortest used area. When $p_1$ and $p_2$ are closed, we repeat this process in $p_3$ and $p_4$. \\

\noindent
\textbf{Assigned Covered Area of Open Primary Shelves:}
In analyzing potential conflicts, it will be useful to know how the assigned covered area of two open shelves packed in alternating fashion compares to the longest used length of either at any given time.

\begin{lemma}
\label{lemma:lu}
Let $S^2$ be a set of two primary shelves with each shelf having height $h_0$, and packing ratio $r_0$. Let both shelves be open and packed with small squares as stated previously and assigned additional area from $b_0$ and $b_{1+}$. Let $\ell_u$ be the length of the longest used length. Then the covered area assigned to $S^2$ is at least $h_0 \ell_u$.
\end{lemma}
\begin{proof}
Let $p_1$ and $p_2$ be two shelves with used lengths $\ell_1$ and $\ell_2$ respectively. We will show that when $\ell_1 = \ell_2 = \ell_u = 0$, adding a square or vertical shelf to one of them gives a covered area of at least $h_0 \ell_u$. It follows that when $\ell_1 \neq \ell_2$, adding to the shorter shelf will not break this invariant.

Without loss of generality, let the new square be added to $p_1$. If this new square adds a vertical column, the analyses is simple. The covered area assigned to $p_1$ will be $h_0r_0 \ell_u = \frac{h_0 \ell_u}{2}$, but the largest vertical column has height $h_1$. So in this case $\ell_u \leq \ell_1 + h_1 = \ell_2 + h_0r_0$. By assigning, $(h_0r_0)^2$ from $b_0$ to $p_2$, we see that $p_2$ contributes an area of $h_0r_0 \ell_2 + (h_0r_0)^2 = h_0r_0(\ell_2 + h_0r_0) \geq h_0r_0\ell_u$. Adding the shelves together gives a total covered area of at least $h_0 \ell_u$.

On the other hand, suppose the new square added to $p_1$ is in $c_0$. We'll call this new square $Q$ with height $h_Q$. In this case, we can split the area of $Q$ among the two shelves and use $b_0$ to fill in the rest. Because $h_Q > h_0r_0$, we can divide it into four parts: $(h_0r_0)^2$, $h_0r_0(h_Q - h_0r_0)$, $h_0r_0(h_Q - h_0r_0)$ and $(h_Q - h_0r_0)^2$. The first two parts can be assigned to $p_1$, adding the necessary $h_0r_0h_Q$ to its area. The third part along with $(h_0r_0)^2$ from $b_0$ can be assigned to $p_2$, adding $h_0r_0h_Q$ to its area as well. Since $\ell_u = \ell_1 + h_Q$ this brings the assigned covered area to $h_0 \ell_u$.
\end{proof}

\noindent
\textbf{Assigned Covered Area of Closed Primary Shelves:  }
Now we consider what happens when two alternately packed primary shelves are closed. In this algorithm, when we receive a small square which fits into neither shelf, we close both shelves.

\begin{lemma}
\label{lemma:2closed}
A pair of closed primary shelves that were alternately packed has been assigned a covered area of at least $2(\ell h_0r_0 - (h_0r_0)^2)$.
\end{lemma}
\begin{proof}
Let $Q$ be the new square which fits neither shelf and let $h_Q$ be it's height if $Q \in c_0$ or the height of the vertical shelf if $Q \in c_{1+}$. Then the unused regions of both shelves must be shorter than $h_Q$. For one shelf, we can use lemma~\ref{lemma:vert} to ensure the wasted space is less than $(h_0r_0)^2$. For the other shelf, we can assign $(h_0r_0)^2$ from $b_0$ to ensure that the wasted space is at most $h_Qh_0r_0 - (h_0r_0)^2 \leq h_0^2r_0 - (h_0r_0)^2 = (h_0r_0)^2$. Then the total assigned covered area is at least $2(\ell h_0r_0 - (h_0r_0)^2)$.
\end{proof}

\noindent
\textbf{Small and Very Small Squares:  }
We now consider inputs containing only small and very small squares and show that these inputs can be packed.

\begin{lemma}
\label{lemma:smallvery}
Any set of small and very small squares with total area $\leq 3/8$ can be packed into the container.
\end{lemma}
\begin{proof} 
By lemma~\ref{lemma:2closed}, $p_1$ and $p_2$ are assigned a combined covered area of at least $2(1 \cdot 1/4 \cdot 1/2-(1/4 \cdot 1/2)^2) = 7/32$ when they are closed. It follows then that the remaining input of at most $5/32$ should not close $p_3$ and $p_4$. Assume for the sake of contradiction that some small square $Q$ with height $h_Q$ closes these shelves. 

Since we are now packing $p_3$, we know that the second half of $b_0$ has been packed and we can use it to balance the two shelves as in lemma~\ref{lemma:lu}. Then the total area assigned to the shelves plus the area of $Q$ must be at least $0.692h_0r_0 + 0.706h_0r_0 - 2 \cdot h_Qh_0r_0 + h_Q^2 = 1.398/8 - h_Q/4 + h_Q^2$, which is greater than $5/32$, a contradiction. \\
\end{proof}

\noindent
\textbf{Small, Very Small, and Large Squares:  }

\begin{lemma}
\label{lemma:smallverylarge}
Any set of small, very small, and large squares with total area $\leq 3/8$ can be packed into the container.
\end{lemma}
\begin{proof} 
The first concern in this case is that small square packed into $p_1$ and $p_2$ will overlap the region reserved for the largest remaining square. To show that this cannot happen, we consider the sum of two lengths: the length $\ell_u$ of the longest used area among $p_1$ and $p_2$ and the length $\ell_G$ which is equal to the height of the largest possible square remaining in the input. Using Claim $3$, this sum can be expressed as $\ell_u + \ell_G = \ell_u + \sqrt{3/8 - h_0 \ell_u} = \ell_u + \sqrt{3/8 - \ell_u/4}$ which is $\leq 1$ for all $\ell_u \leq 1/2$. When $\ell_u > 1/2$, the total area of squares received is $> 1/8$ by Lemma~\ref{lemma:lu} and large squares are no longer possible.

The second concern is that a large square may interfere with the packing of small squares into the primary shelves. However, we merely consider all four primary shelves to be packed at a density greater than $1/2$ for a used length equal to the height of the large square and all previous lemmas hold.
\end{proof}

\subsection{Combining All Four Size Classes}\label{sec:combine}
In this final subsection, we will show that our algorithm handles inputs containing any combination of size classes. \\

\noindent
\textbf{Additional Terminology:  }
In considering the space which is shared by small and medium shelves we introduce the notion of partially closed shelves. We call a shelf \textit{closed to medium squares} if we have received a medium square which does not fit into the shelf. We call a pair of shelves \textit{closed to small squares} if we have received a small square which does not fit into either shelf. In addition, we consider any pair of shelves closed to small squares to also be closed to medium squares. 

The surprising fact about using the bottom half as both a pair of small shelves and a medium shelf simultaneously is that many previous lemmas still hold.

\begin{lemma}
\label{lemma:bottom}
When packing medium, small, and very small squares with total area $\leq 3/8$, previous guarantees for open and closed shelves still apply.
\end{lemma}
\begin{proof} 
For both sizes, we have a guarantee for open shelves that half of their used area is covered. This holds until the point at which some square $Q$ would cause these two used areas to collide.

If $Q$ is medium, we say that the bottom half is closed to medium squares and because both used areas have the same minimum density of $1/2$, we can still apply Corollary~\ref{cor:MM}. If $Q$ is small we say that the bottom half is closed to small squares and by the same argument, we can apply Lemma~\ref{lemma:2closed}.
\end{proof}

\begin{corollary}
\label{cor:allvslarge}
Packing medium, small, and very small squares together into the bottom half will not overlap the area reserved for the largest possible remaining square.
\end{corollary}
\begin{proof} 
The proof follows from the previous lemma. \\
\end{proof}

\begin{lemma}
\label{lemma:all}
Large, medium, small, and very small squares with total area $\leq 3/8$, can be packed into a unit square container.
\end{lemma}
\begin{proof} 
Using Corollary~\ref{cor:allvslarge} and Lemmas~\ref{lemma:medlarge} and~\ref{lemma:smallverylarge} we can see that a proof of this lemma for inputs without large squares will also apply to inputs with large squares. So for simplicity, we consider inputs without large squares.

To show that such inputs can be packed into our container, we focus on the final square which is received by our algorithm. If the final square received is small or very small, then Lemma~\ref{lemma:smallvery} applies. If the final square is medium, we consider two cases. \\

\noindent
\textbf{Case 1:}
The top half (excluding the buffer region) is packed with only medium squares, then Lemma~\ref{lemma:med} applies. \\

\noindent
\textbf{Case 2:}
The top half (excluding the buffer region) is packed with small/very small and medium squares. In this case, the bottom half is closed to small squares and Lemma~\ref{lemma:bottom} guarantees that the bottom half has been assigned a covered area of at least $7/32$ and only $5/32$ remains for the top half. We assume for the sake of contradiction that the final medium square $Q$ with height $h_Q$ does not fit into the top shelf which has length greater than $0.69$. Then area assigned to the top half must be at least $1/4(0.69 - h_Q) + h_Q^2$ which is greater than $5/32$, a contradiction. \\
\end{proof}

\begin{theorem}
\label{theorem:squares}
Any set of squares with total area at most $\leq 3/8$, which is received in an online fashion, can be packed into a unit square container.
\end{theorem}
\begin{proof} 
The proof follows from all previous lemmas.
\end{proof}

\section{Extensions}
We are currently extending this work to the more general case of packing rectangles into a unit square.

\end{document}